\documentclass[letterpaper,twocolumn,10pt]{article}
\usepackage{usenix-2020-09}

\usepackage{tikz}
\usepackage{amsmath}
\usepackage{thmtools,thm-restate}

\usepackage{tabularx}

\usepackage{filecontents}

\usepackage{xspace}
\usepackage{diagbox}
\usepackage{comment}
\usepackage{authblk}
\usepackage{times}
\usepackage{amsmath,amsfonts,amssymb,amsthm}
\usepackage{graphicx}
\newtheorem{lemma}{Lemma}
\newtheorem{theorem}{Theorem}

\newtheorem{definition}{Definition}

\usepackage[normalem]{ulem}
\usepackage{mathtools}
\usepackage{authblk}
\usepackage{times}
\usepackage{amsmath, amsthm, graphicx, comment, xspace, amssymb}
\usepackage{bm}

\graphicspath{{./imgs/}}

\theoremstyle{definition}

\newcommand{\Minerva}{\textsc{Minerva}\xspace}
\newcommand{\Providence}{\textsc{Providence}\xspace}
\newcommand{\B}{{{B2}}\xspace}
\newcommand{\R}{{{R2}}\xspace}
\newcommand{\BRAVO}{\textsc{Bravo}\xspace}

\newcommand{\Bravo}{\textsc{Bravo}\xspace}

\begin{document}

\date{}

\title{\Large \bf \Providence: a Flexible Round-by-Round Risk-Limiting Audit}
\author{
{\rm Oliver Broadrick$^1$\thanks{odbroadrick@gmail.com}, Poorvi L. Vora$^1$, and Filip Zag{\'o}rski$^{23}$}\\
$^1$Department of Computer Science, The George Washington University\thanks{Authors supported in part by NSF Award 2015253}\\
$^2$University of Wroclaw\\
$^3$Votifica
} 
\maketitle

\begin{abstract}
A Risk-Limiting Audit (RLA) is a statistical election tabulation audit with a rigorous error guarantee. We present ballot polling RLA \Providence, an audit with the efficiency of \Minerva and flexibility of \BRAVO. We prove that \Providence is risk-limiting in the presence of an adversary who can choose subsequent round sizes given knowledge of previous samples. We describe a measure of audit workload as a function of the number of rounds, precincts touched, and ballots drawn. We quantify the problem of obtaining a misleading audit sample when rounds are too small, demonstrating the importance of the resulting constraint on audit planning. We present simulation results demonstrating the superiority of \Providence using these measures and describing an approach to planning audit round schedules. 

We describe the use of \Providence by the Rhode Island Board of Elections in a tabulation audit of the 2021 election. Our implementation of \Providence and audit planning tools in the open source R2B2 library should be useful to the states of Georgia and Pennsylvania, which are planning pre-certification ballot polling RLAs for the 2022 general election. 
\end{abstract}

\section{Introduction}
\label{sec:intro}
It is well-known that electronic voting systems are vulnerable to software errors and manipulation which may be undetected. Errors and/or manipulation may not always change an election outcome, but we would want to know when they do. {\em Software independent} voting systems \cite{SI-Wack,rivest2008notion} are ones where an undetected change in the software cannot lead to an undetected change in the election outcome. {\em Evidence-based elections} \cite{evidence-based} use software independent systems to produce trustworthy evidence of outcome correctness; incorrect outcomes are detected with high probability when the evidence is examined. One approach to evidence-based elections is to use voter-verified paper ballots, store them securely, and perform public audits---a compliance audit to determine whether the ballots were stored securely; and a rigorous tabulation audit, known as a risk-limiting audit (RLA) \cite{RLA}, to determine whether the outcome is correctly computed from the stored ballots. Many US states have had pilot RLA programs. Additionally, some states allow RLAs to be used towards audit requirements, and some states require RLAs before elections can be certified. 

We propose the \Providence audit, a new approach to the ballot polling RLA, and propose a new model for the work load of an election. We show that \Providence is superior to the popular ballot polling RLA \Bravo for real elections, and describe the use of our open source implementation by the Rhode Island Board of Elections for an audit of their 2021 elections. Our open-source implementations of \Providence and our audit planning tools are likely to be useful later this year; ballot polling audits are expected to be used as pre-certification RLAs in at least one statewide contest in both Georgia and Pennsylvania for the 2022 general elections in the US.   

\subsection{Background}
Ballot comparison RLAs require the fewest ballots of all known RLA approaches. On the other hand, they require the use of special election technology and are not always feasible. Ballot polling RLAs require a larger number of ballots, but are more feasible because they do not require any additional functionality of the voting system. What is needed is a complete ballot manifest (a list of ballot storage containers and the number of ballots in each) which enables the creation of a well defined list of the ballots and their locations (the fifth ballot in box number 20, for example). 

A ballot polling audit begins when a {\em round} \cite{usenix_minerva} of multiple randomly-chosen ballots is drawn. A risk measure is then computed to determine whether (a) the audit ends in success (the election outcome is declared correct) or (b) another round should be drawn. Election officials would typically decide to perform a full manual hand count if the audit does not stop in spite of drawing a large number of ballots, typically over multiple rounds. Ballot polling audits have been used in a number of US state pilots (California, Georgia, Indiana, Michigan, Ohio, Pennsylvania, Virginia and elsewhere) and in real statewide audits (Georgia, Virgina) \cite{vv_audits} as well as in audits of smaller jurisdictions, such as Montgomery County, Ohio \cite{usenix_minerva}. 

A {\em round-by-round (R2)} audit is one where the decision of whether to draw more ballots or not is taken after drawing a round of ballots; typically hundreds or thousands or tens of thousands of ballots in statewide elections. A {\em ballot-by-ballot (B2)} audit is the special case of round size one---when the decision is made after each ballot is drawn. The popular \BRAVO audit requires the smallest expected number of ballots when the announced tally of the election is correct, and stopping decisions are taken a ballot at a time (that is, when it is used as a B2 audit). Election officials typically draw ballots in large round sizes, and \BRAVO needs to be modified for use in this manner. For use as an R2 audit, the \BRAVO stopping condition can be applied once at the end of each round (End-of Round (EoR)), or retroactively after each ballot drawn if ballot order is retained (Selection-Ordered (SO)). SO \BRAVO is closer to the original B2 \BRAVO, and requires fewer ballots on average than EoR \BRAVO. But it requires the additional effort of tracking the order of ballots. 

Zag\'{o}rski {\em et al.} propose ballot polling RLA \Minerva \cite{usenix_minerva}, which does not need ballot order and relies only on sample and round tallies. They prove that it requires fewer ballots than EoR \BRAVO when both audits have the same pre-determined round schedule and the true tally is as announced. 
They also present first-round simulations demonstrating that \Minerva draws fewer ballots than SO \BRAVO in the first round for large first round sizes when the true tally is as announced. 
Broadrick {\em et al.} provide further simulations that show \Minerva requires fewer ballots over multiple rounds and for lower stopping probability \cite{simulations}, though the improvement from using \Minerva over either version of \BRAVO decreases with round size. 

\subsection{Open Questions in the Literature}
A major limitation of \Minerva is that one needs to determine the round schedule before the audit begins because \Minerva has not been shown to be risk-limiting if an adversary can choose subsequent round sizes after knowing the sample drawn. \BRAVO, on the other hand, is not limited in this manner. This allows \BRAVO audits the flexibility of choosing smaller subsequent round sizes if the sample drawn so far is a ``good'' sample. An open question is whether a ballot polling RLA exists with the efficiency of \Minerva and this flexibility of \BRAVO.

A major limitation of our understanding of the ballot polling problem as a community is that we use the number of ballots drawn or values proportional to this number \cite{mclaughlin_thesis,bernhard-diss,RI-report} as measures of the workload of an audit. If this were a correct measure of the workload of an audit, we would want to use B2 audits (round size is one) and make decisions about stopping the audit after drawing each ballot, because this leads to the smallest expected number of ballots. Election officials, on the other hand, greatly prefer drawing many ballots at once. This preference is likely due to the following. 
\begin{description}
\item Firstly, each round has an overhead workload as well, including setting up the round and communicating among the various localities involved in conducting the audit (for example, audits of statewide contests involve the drawing of ballots at county offices where the ballots are stored). 
\item Secondly, there is an overhead to finding a storage box and unsealing it. For large round sizes, multiple ballots may be drawn at once from a box, and the number of boxes retrieved is smaller than the number of ballots (storage boxes commonly contain many hundreds of ballots each). For smaller round sizes, the number of times a box is retrieved would be roughly identical to the number of ballots drawn, as it is unlikely that a single box will hold multiple ballots from the sample. 
\item Finally, in the current environment of misinformation, election officials would want to ensure that the probability of a misleading audit sample (falsely indicating that the loser won) is very small, which implies that round sizes should be large. 
\end{description}
Thus the workload of an audit is not simply a linear (or affine) function of the number of ballots drawn. Relatedly, an optimal round schedule is not completely determined by the expected number of ballots drawn. It depends on other variables as well. The consideration of all these variables is necessary while planning an audit. 

\subsection{Our Contributions}
We present \Providence, and provide the following:
\begin{enumerate}
\item Proof that \Providence is an RLA and resistant to an adversary who can choose subsequent round sizes with knowledge of previous samples.
\item Simulations of \Providence, \Minerva, SO \BRAVO, and EoR \BRAVO which show that \Providence uses number of ballots similar to those of \Minerva, both fewer than either version of \BRAVO.
\item Results and analysis from the use of \Providence in a pilot audit in Rhode Island.
\item A model of workload that includes the overhead effort of each round and the overhead effort of retrieving a storage unit of ballots; simulations that illustrate the use of this model to compare the different types of ballot polling audits and to plan an audit with minimal workload.
\item An analysis of round size as a function of the maximum acceptable probability of a misleading audit sample.
\item Open source implementation of \Providence and audit planning tools. 
\end{enumerate}

Our results demonstrate the superiority of \Providence over the other audits. Our work may be used by election officials to plan ballot polling audits, including in Georgia and Pennsylvania in 2022. 

\subsection{Organization} 
Section \ref{sec:related} describes related work. Section \ref{sec:prov} describes the \Providence audit, section \ref{sec:sims} the simulations comparing the number of ballots drawn using various ballot polling audits and section \ref{sec:pilot} the use of \Providence in an audit carried out by the Board of Elections of Rhode Island. Section \ref{sec:workload} presents our workload model and describes its use for a ballot polling audit using details of the 2020 US Presidential election in the state of Virginia. Our conclusions, the availability of an audit implementation and acknowledgements may be found in sections \ref{sec:conc}, \ref{sec:avail} and \ref{sec:ack} respectively.

\section{Related work}
\label{sec:related}
Bernhard provides a good description of the RLA and its assumptions, and also describes the process on the ground \cite{bernhard-sok}. 

The \BRAVO audit \cite{bravo} is the most popular ballot polling audit. When ballots are sampled one at a time, it is the audit with the smallest expected number of ballots drawn. 

The \Minerva audit \cite{usenix_minerva,arxiv_athena} was developed for use with large first round sizes, and has been proven to be risk limiting when the round schedule for the audit is fixed before any ballots are drawn. First-round sizes for a stopping probability of $0.9$ when the announced tally is correct have been shown to be smaller than those for EoR and SO \BRAVO for a wide range of margins; simulations \cite{arxiv_athena} support these observations. Additional simulations \cite{simulations} have shown that \Minerva requires fewer ballots than EoR and SO \BRAVO over multiple rounds and for smaller stopping probability. As expected, the advantage of \Minerva decreases for smaller stopping probability (smaller round sizes) as such round schedules approach the B2 round schedule (1,1,1,\ldots) for which \BRAVO is known to be most efficient.

Ballot polling audit simulations provide a means of educating the public and election officials \cite{dice} and to understand audit properties \cite{mclaughlin_thesis,simulations_house, blom_IRV, DBLP:conf/evoteid/HuangRSTV20}. There is work measuring the amount of time taken to examine a single ballot \cite{RI-report}. 
Simple workload estimates may be obtained by using the number of ballots drawn \cite{bernoulli-ballot-polling}, a more thorough workload estimation model includes the time taken to access individual ballots\cite{bernhard-diss}. 

We now summarize the model drawing largely from the notation and terminology of \cite{usenix_minerva,arxiv_athena,simulations,bravo}. The model is related work and not claimed to be original to this work. 

An audit $\mathcal{A}$ is a function that takes as input the sample of ballots and outputs either (1) \emph{Correct: stop the audit} or (2) \emph{Undetermined: sample more ballots}.
\BRAVO and \Minerva are modeled as binary hypothesis tests where the null hypothesis $H_0$ corresponds to a tied election and the alternative hypothesis $H_a$ to an election tally as announced. 
(When the number of ballots is odd, $H_0$ corresponds to the announced loser winning by one ballot.)
Thus the null hypothesis is the outcome distinct from the announced one which is most difficult to detect; the probability of failing to detect it, given that the null hypothesis is true, is the worst case such probability and should be below the risk limit \cite{Bayesian-RLA}.

\begin{definition}[Risk Limiting Audit ($\alpha$-RLA)]
An audit $\mathcal{A}$ is a Risk Limiting Audit with 
risk limit $\alpha$ iff for sample $X$
$$
Pr[\mathcal{A}(X) 
= \text{Correct} \,|\, H_0]\le \alpha
$$
\end{definition}

The stopping conditions of \BRAVO and \Minerva rely on the following ratios.

\begin{definition}[\BRAVO Ratio] \label{def:bravo-ratio} The \BRAVO audit uses the ratio $\sigma$. Consider a sample size of $n$ ballots with $k$ for the reported winner. The proportion of ballots for the reported winner under the alternative hypothesis and null hypothesis are $p_a$ and $p_0$ respectively.
\begin{equation}
    \sigma(k, p_a, p_0, n) \triangleq \frac{p_a^{k} (1-p_a)^{n-k}}{p_0^{k} (1-p_0)^{n-k}} 
    \label{eqn:bravoratio}
\end{equation}
\end{definition}

In \BRAVO, $p_0=\frac{1}{2}$. A \BRAVO audit outputs correct if and only if
$$\sigma(k,p_a,\frac{1}{2},n)\ge \frac{1}{\alpha}.$$

It is easy to see that the ratio $\sigma$ is the 
likelihood ratio:
$$
\frac{Pr[K=k|H_a,n]}{Pr[K=k|H_0,n]}= \frac{\binom{n}{k}p_a^{k} (1-p_a)^{n-k}}{\binom{n}{k}(\frac{1}{2})^n} =\sigma(k, p_a, \frac{1}{2}, n)
$$

Where \BRAVO uses the ratio of the values of the probability distribution functions, \Minerva uses the ratio of their \emph{tails}. Now it becomes useful to have shorthand for a sequence of cumulative round sizes and the corresponding sequence
of cumulative winner ballot tallies.
We use:
$$\bm{k_j}\triangleq(k_1,k_2,\ldots,k_j) \quad\text{and}\quad \bm{n_j}\triangleq(n_1,n_2,\ldots,n_j)$$

\begin{definition}[\Minerva Ratio] \label{def:minerva_ratio} The \R \Minerva audit uses the ratio $\tau_j$. We use cumulative round sizes $\bm{n_j}$, with corresponding $\bm{k_j}$ ballots for the reported winner in each round. The proportion of ballots for the reported winner under the alternative hypothesis and null hypothesis are $p_a$ and $p_0$ respectively.
         \begin{equation}
         \begin{aligned}
             \label{eqn:tau}
                 \tau_{j}(k_{j}, p_a,p_0, \bm{n_j}, \alpha )  \triangleq\\
                 \frac{Pr[K_{j} \geq k_{j} \wedge \forall_{i < j} ({\mathcal{A}}(X_i) ~\neq \text{Correct}) \mid H_a, \bm{n_j}]}{Pr[K_{j} \geq k_{j} \wedge \forall_{i < j} ({\mathcal{A}}(X_i) ~\neq \text{Correct}) \mid H_0, \bm{n_j}]}
         \end{aligned}
         \end{equation}
\end{definition}

\section{\Providence}
\label{sec:prov}
In this section we introduce the stopping condition of \Providence and prove some properties.

Recall that the proof that the \Minerva audit is risk-limiting assumes that the round schedule of \Minerva is predetermined and that, in particular, an adversarial auditor cannot determine the next round size after drawing a sample. This presents difficulties because a non-adversarial election official might want to draw a small next round if the current sample comes close to satisfying the risk limit. Because the \Minerva round size is predetermined, however, the election official would be required to draw a larger round size than necessary for the sample. Conversely, if the current sample is not at all close to satisfying the risk limit, it would be advantageous to draw a larger round than the predetermined round size. 

The \Providence audit is risk-limiting even if an adversarial auditor determines round sizes after drawing the sample, and next round size computations may use knowledge of the current sample. 

\subsection{Definition}
\label{sec:prov_def}
\begin{definition}[$(\alpha,p_a, p_0,k_{j-1},n_{j-1},n_j)$-\Providence]
    \label{def:minervatwo}
    For cumulative round size $n_j$ for round $j$ and a cumulative $k_j$ ballots for the reported winner found in round $j$, the \R \Providence stopping rule for the $j^{th}$ round is:
$$
\mathcal{A}(X_{j})=  \left\{ \begin{array}{ll} \text{Correct} ~~~~ \omega_{j}(k_{j}, k_{j-1}, p_a, p_0, n_j, n_{j-1}) \geq \frac{1}{\alpha}\\
        Undetermined ~~else \\
    \end{array}
    \right .
$$
where $\omega _{1}\triangleq \tau_{1}$ and for $j\ge 2$, we define $\omega _{j}$ as follows:
\begin{equation}
    \begin{aligned}
    \omega_{j}(k_{j}, k_{j-1}, p_a, p_0, n_{j}, n_{j-1})
    \triangleq\\
    \sigma(k_{j-1},p_a,p_0,n_{j-1})\cdot \tau_1(k_{j}-k_{j-1},p_a,p_0,n_j-n_{j-1})
    \end{aligned}
\end{equation}
\end{definition}

Notice that \Providence requires the computation of $\tau_j$ for $j=1$ and no other values of $j$. The value of $\tau_1$ is simply the ratio of the the tails of the binomial distributions for the two hypotheses and can be fairly efficiently computed. The computation of $\tau_j$ for $j \geq 2$, as required in \Minerva, relies on the convolution of two probability distribution functions and is hence computationally considerably more expensive. 

Notice also that \Providence and \Minerva are identical for $j=1$. 

\subsection{Risk-Limiting Property: Proof}
\label{sec:proof}
We now prove that \Providence is risk-limiting using lemmas from basic algebra in Appendix~\ref{sec:proofs}.

\begin{theorem}
\label{thm:minerva2_is_rla_new}
An $(\alpha,p_a, p_0,k_{j-1},n_{j-1},n_j)$-\Providence audit is an
$\alpha$-RLA.
\end{theorem}
\begin{proof}
Let $\mathcal{A}=(\alpha,p_a, p_0,k_{j-1},n_{j-1},n_j)$-\Providence.
Let $\bm{n_j}$ be the cumulative round sizes used in this
audit, with corresponding cumulative tallies of
ballots for the reported winner $\bm{k_j}$.
For round $j=1$, by Definitions \ref{def:minervatwo}
and \ref{def:minerva_ratio}, we see that
the $\mathcal{A}=\text{Correct}$ (the audit stops) only when
$$
\tau_1(k_{1},p_a,p_0,n_1)\\
=\frac{Pr[K_{1} \geq k_{1} \mid H_a, n_1]}{Pr[K_{1} \geq k_{1} \mid H_0, n_1]}
\ge \frac{1}{\alpha}.
$$
By Lemma \ref{lemma:minerva2_kmin_exists}, we see that this
is equivalent to the following:
$$
\frac{Pr[K_{1} \geq k_{min,1} \mid H_a, n_1]}{Pr[K_{1} \geq k_{min, 1} \mid H_0, n_1]}
\ge \frac{1}{\alpha}.
$$
where $k_{min,1} = k^{p_a, p_0, \alpha, 0}_{min, 1, 0, n_1}$ (see Lemma \ref{lemma:minerva2_kmin_exists} and Definition \ref{def:kmin}). 

For any round $j\ge 2$, by Definition \ref{def:minervatwo}
and Lemma \ref{lemma:minerva2_kmin_exists},
$\mathcal{A}=\text{Correct}$ (the audit stops) if and only if
\begin{equation*}
\begin{aligned}
\omega_{j}(k_{j}, k_{j-1}, p_a, p_0, n_{j}, n_{j-1}, \alpha )\triangleq\\
\sigma(k_{j-1},p_a,p_0,n_{j-1})\cdot \tau_1(k_{j}-k_{j-1},p_a,p_0,n_j-n_{j-1})
\ge \frac{1}{\alpha}.
\end{aligned}
\end{equation*}
By Lemma \ref{lemma:any_ratio_is_sigma_simple}
and Definition \ref{def:minerva_ratio}, this is equivalent to
$$
\frac{\Pr[K_{j-1} = {k_{j-1}} \mid H_a, n_{j-1}] Pr[K_{j} \ge k_{j} \mid {k_{j-1}}, H_a, n_{j-1}, n_{j}]}{\Pr[K_{j-1} = {k_{j-1}} \mid H_0, n_{j-1}] Pr[K_{j} \ge k_{j} \mid {k_{j-1}}, H_0, n_{j-1}, n_{j}]}
$$$$\ge \frac{1}{\alpha}.
$$
By Lemma~\ref{lemma:minerva2_kmin_exists} and Definition~\ref{def:minervatwo},
we see that there exists a $k_{min, j} = k^{p_a, p_0, \alpha, k_{j-1}}_{min, j, n_{j-1}, n_j}  \leq k_j$ 
for which
$$
\frac{\Pr[K_{j-1} = {k_{j-1}} \mid H_a, n_{j-1}] Pr[K_{j} \ge k_{j} \mid {k_{j-1}}, H_a, n_{j-1}, n_{j}]}{\Pr[K_{j-1} = {k_{j-1}} \mid H_0, n_{j-1}] Pr[K_{j} \ge k_{j} \mid {k_{j-1}}, H_0, n_{j-1}, n_{j}]}\ge
$$
$$
\frac{\Pr[K_{j-1} = {k_{j-1}} \mid H_a, n_{j-1}] Pr[K_{j} \ge k_{min, j} \mid {k_{j-1}}, H_a, n_{j-1}, n_{j}]}{\Pr[K_{j-1} = {k_{j-1}} \mid H_0, n_{j-1}] Pr[K_{j} \ge k_{min, j} \mid {k_{j-1}}, H_0, n_{j-1}, n_{j}]} 
$$$$
\ge 
\frac{1}{\alpha}
$$
The above may be rewritten as
\begin{equation*}
\begin{aligned}
\sum_{{k} = k_{min, j}}^{n_j} Pr[(K_{j} , K_{j-1}) = (k, k_{j-1}) \mid H_0, n_{j-1}, n_{j}] \leq \\
\alpha \sum_{{k} = k_{min, j}}^{n_j} Pr[(K_{j} , K_{j-1}) = (k, k_{j-1}) \mid H_a, n_{j-1}, n_{j}]
\end{aligned}
\end{equation*}
The left hand side above is the probability of stopping in the $j^{th}$ round and $K_{j-1} = k_{j-1}$, given the null hypothesis, which is smaller than $alpha$ times the same probability given the alternate hypothesis. Summing both sides over all values of $K_{j-1} < k_{min, j-1}$ gives us a similar relationship between the probabilities of stopping in round $j$ (given the null and alternate hypotheses respectively). Note here that the relationship holds even if the values of $n_{j}$ depend on $k_{j-1}$. When both sides of the inequality are further summed over all rounds, we get:  

$$
Pr[\mathcal{A}=\text{Correct} \mid H_0]
\le
\alpha Pr[\mathcal{A}=\text{Correct} \mid H_a]
$$
Finally, because the total probability of stopping the audit under
the alternative hypothesis is not greater than 1, we get
$$
Pr[\mathcal{A}=\text{Correct} \mid H_0] \le
\alpha.
$$
\end{proof}

\subsection{Resistance to an adversary choosing round sizes}
\label{sec:adversary}


\Minerva was proven to be a risk-limiting audit for a predetermined round schedule.
As explained earlier, it is not clear that \Minerva is risk-limiting if an adversary can 
adaptively select the round schedule as the audit proceeds. In this section we prove that \Providence does not have this problem, and is risk-limiting even when the adversary can choose next round sizes based on knowledge of the current sample. 

\begin{definition}[Strategy-Proof RLA]
 An audit $\mathcal{A}$ is a Strategy-Proof Risk Limiting Audit with risk limit $\alpha$ iff for
 all strategies of selecting round schedule and for sample $X$ 
 \[
  \Pr\left[\mathcal{A}(X) = \text{Correct} | H_0\right] \leq \alpha.
 \]

\end{definition}

\begin{lemma}
\Providence is a Strategy-Proof RLA. 
\end{lemma}
\begin{proof}
This property follows from Theorem~\ref{thm:minerva2_is_rla_new} and Lemma~\ref{lemma:markov}. Note that, as described in section \ref{sec:proof}, the proof of the risk-limiting nature of the audit does not rely on round sizes $n_j$ being identical for all values of $k_{j-1}$. 
\end{proof}

To illustrate the practical implication of this property, we consider a toy example: an RLA of a two-candidate contest with margin $0.01$ and risk limit $0.1$. 
Suppose we wish to achieve a conditional stopping probability $0.9$ in each round of the audit. For \Providence, we can compute a new round size for each round based on the previous samples. For \Minerva, however, we would have a predetermined round schedule. We use the default \Minerva round schedule of audit software Arlo \cite{arlo} (used by many states performing an RLA), which is is $[x, 2.5x, 6.25x, ...]$; that is, the next marginal round size is $1.5$ times the current one. This multiplier of $1.5$ is known to give, over a wide range of margins, a probability of stopping roughly $0.9$ in the second round if the first round size has probability of stopping $0.9$. 

Both the audits of our toy example therefore begin with a first round size of $17,272$ with a $0.9$ probability of stopping, and both will stop in the first round if the sample contains at least $8,725$ ballots for the winner. We now consider two cases for which the audit proceeds to a second round. 
\begin{description}
\item In one case there are $8,724$ votes for the winner in the sample, just one fewer than the minimum needed to meet the risk limit. In the \Minerva audit, we are already committed to a second round size of $43,180$ which, given the nearly-passing sample of the first round is higher than necessary, achieving a stopping probability in the second round of $.954$. The \Providence audit samples more than $9,000$ fewer ballots with a round size of $34,078$, achieving the desired $0.9$ probability of stopping.
\item In a less lucky sample, the winner recieves $8,637$ ballots, few more than the loser recieves. In the \Minerva audit, we again have to use a second round size of $43,180$, but now this round size only achieves a $0.727$ probability of stopping, significantly less than the desired $0.9$. Again, the \Providence audit can scale up the second round size according to the first sample and achieve the desired $0.9$ probability of stopping with $58,007$ ballots.
\end{description}

\subsection{Efficiency}
\begin{lemma}
\label{lem:efficiency}
For any risk-limit $\alpha \in (0, 1)$, for any margin
and for any round schedule $[n_1, \ldots, n_j]$, 
the \Providence RLA is more efficient than EoR \BRAVO.
\end{lemma}
\begin{proof} Appendix~\ref{sec:proofs}\end{proof}

\section{\Providence Audit Simulations}
\label{sec:sims}
We use simulations to provide additional evidence for our theoretical claims regarding \Providence and to gain insight into audit behavior. As done in \cite{simulations}, we use margins from the 2020 US Presidential election---state-wide pairwise margins of $0.05$ or larger between the two leading candidates. Narrower margins are computationally expensive, especially for the simulations of tied elections, which, by design, have a low probability of stopping and hence quickly increase in sample size. We use the simulator in the R2B2 software library\cite{r2b2}. For each margin, we perform $10^4$ \Providence audit trials each on a tied election (hypothesis $H_0$, the null hypothesis) and the election as reported (hypothesis $H_a$, the alternate hypothesis). All trials have risk limit $\alpha = 0.1$, a maximum of $5$ rounds, and a conditional stopping probability of $0.90$ in each round. That is, each next round size is selected to be large enough to give a $0.90$ conditional probability of stopping in that round, assuming the announced tally is correct and given the tally of previous rounds. We use a maximum of five rounds because virtually no audits would progress beyond five rounds given the large conditional probability of stopping. 

In the simulations of \Providence audits of a tied election, the fraction of audits that stop, as shown in Figure~\ref{fig:prov-risk}, is an estimate of maximum risk. For all margins, this estimated maximum risk is less than the risk limit, supporting the claim that \Providence is risk-limiting.

\begin{figure}[h!]
\includegraphics[width=.5\textwidth]{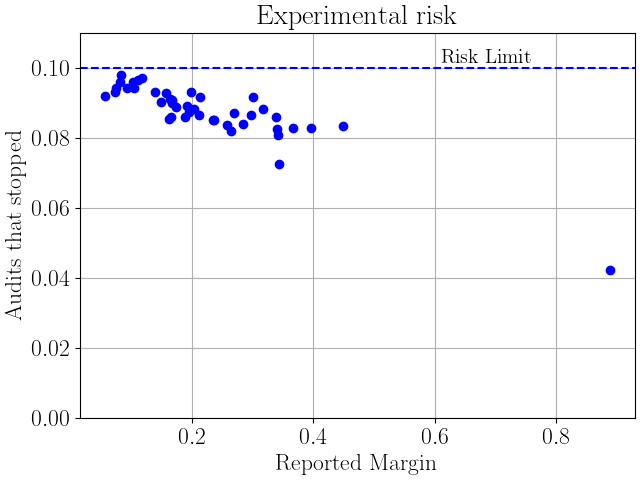}
\caption{The fraction of simulated \Providence audits on tied elections that stopped in any rounds (we performed five rounds at a risk limit of $0.1$) as a function of contest margin. This value is an estimate of the maximum risk of the \Providence audit.}
\label{fig:prov-risk}
\end{figure}

Simulations of audits of the election as reported provide insight into stopping probability and number of ballots drawn when the election is as reported. Figure~\ref{fig:prov-sprob} shows that the stopping probabilities over the first rounds are near and slightly above $0.9$ as expected, since our software chose round sizes to give at least a $0.9$ conditional stopping probability. The values are not as tight around $0.9$ for later rounds because fewer audit trials make it to later rounds, and our experimental probability estimates are not as accurate. 

\begin{figure}[h!]
\includegraphics[width=.5\textwidth]{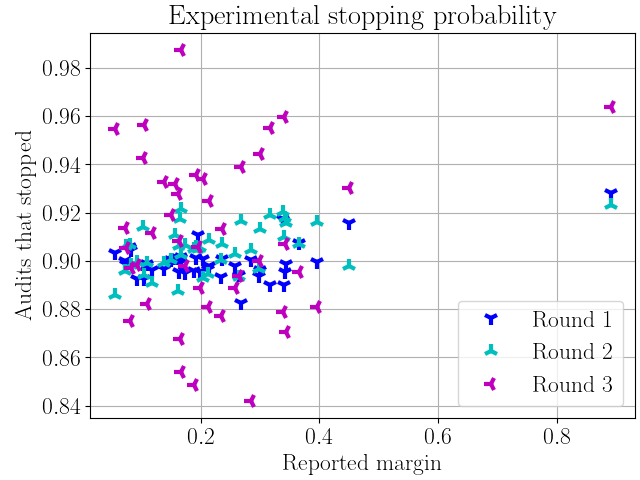}
\caption{The fraction of simulated \Providence audits of the election as reported that stopped for each round as a function of margin. This value is an estimate of the stopping probability conditioned on the sample of the previous round. The average fraction for rounds 1, 2, and 3 is $0.8996$, $0.9052$, and $0.9098$ respectively. We show only the first three rounds since so few audits make it to rounds 4 and 5 (of the order of $10^4 \times (0.1)^3$ and $10^4 \times (0.1)^4$ respectively).}
\label{fig:prov-sprob}
\end{figure}

We now investigate the efficiency of \Providence compared to \Minerva, SO \BRAVO, and EoR \BRAVO by taking a single margin as an example: the 2020 US Presidential election in the state of Texas, with margin $0.057$. We run an additional $10^4$ simulations for each of the three other audits on the same underlying election and on a tied election. Both \BRAVO implementations use a conditional stopping probability of $0.9$ for each round, while \Minerva uses a first round size with stopping probability $0.9$ and a multiplier of $1.5$ to obtain subsequent round sizes. 

\begin{figure}[h!]
\includegraphics[width=.5\textwidth]{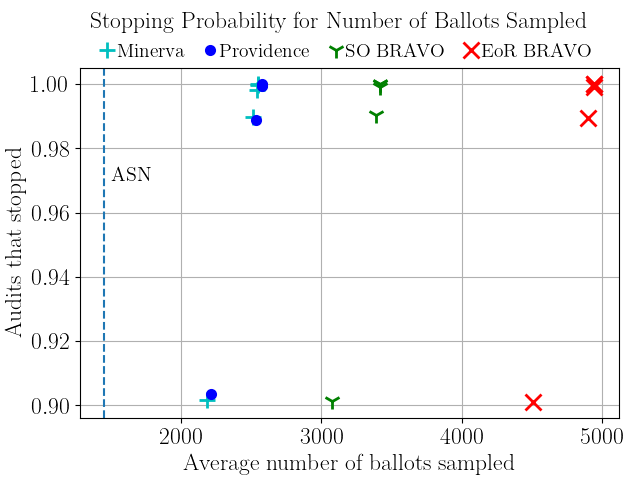}
\caption{For the entire audit, consisting of all five rounds, the fraction of simulated audits that stopped as a function of the average number of ballots drawn for \Providence, \Minerva, EoR \BRAVO, and SO \BRAVO. The average sample number (ASN) for \B \BRAVO is included for context.}
\label{fig:prov-asn}
\end{figure}
Figure~\ref{fig:prov-asn} shows the probability of stopping as a function of the number of ballots sampled, a plot similar to those presented in \cite{simulations}. Points above (higher probability of stopping) and to the left (fewer ballots) represent more efficient audits. As shown, \Providence has comparable efficiency to \Minerva, while both are significantly more efficient than either implementation of \BRAVO. In a contest with a narrow margin (in the 2020 US Presidential election, eight states had margins smaller than $0.03$) the difference in number of ballots sampled could correspond to many days of work which would need to be completed before a certification deadline.

\section{Pilot use}
\label{sec:pilot}
The Rhode Island Board of Elections performed a pilot audit in Providence 
in February 2022. The contest audited was a single yes-or-no question in the November 2021 election: Portsmouth's
Issue 1, "School Construction and Renovation Projects". The question had a reported margin of $0.2567$ and the audit used a risk-limit of $0.10$.

A first round size of $140$ ballots with large probability of stopping ($0.95$) was selected.
Selection order was tracked for the sake of analysis.
As expected, the audit concluded in the first round. The \Providence risk was $0.0418$. Table~\ref{tab:pilot-risks} shows risk measures for the drawn sample using \Providence, \Minerva and \BRAVO (both EoR and SO).

\begin{table}[h!]
\begin{center}
\begin{tabular}{ |c|c|c|c|c| } 
\hline
ballots& \rotatebox{45}{\Providence} & \rotatebox{45}{\Minerva} & \rotatebox{45}{SO \BRAVO} & \rotatebox{45}{EoR \BRAVO} \\
\hline
140 & \bf{0.0418} & \bf{0.0418} & \bf{0.0541} & 0.366 \\
\hline
\end{tabular}
\end{center}
\caption{Risk measures for the drawn first round of $140$ ballots in the Providence, RI pilot audit. Risks in bold meet the risk-limit ($10\%$) and thus correspond to audits that would stop.}
\label{tab:pilot-risks}
\end{table}

Note that the risk measures shown in Table~\ref{tab:pilot-risks} imply that, for the sample obtained in the pilot audit, an EoR \BRAVO audit would not have stopped in the first round, despite the large round size. Further, if the risk limit had been $0.05$ instead of $0.10$, SO \BRAVO also would have required moving on to a second round. 

We can use simulations to better understand typical audit behavior for the margin of this pilot audit and contextualize the results we obtained in the pilot. We run $10^4$ trial audits for several stopping probabilities $p$. Each round size is chosen to give a probability of stopping $p$ assuming the announced tally and given the results of previous rounds. We use the same $0.1$ risk limit and margin of $0.2557$. 

Figure~\ref{fig:pilot_sims} shows the average number of ballots sampled for each value of $p$ in the simulations. The vertical line denotes the stopping probability of the first round size actually chosen in the pilot ($140$ ballots). The large value of $p$ corresponds to a large first round size and a corresponding large value of average number of ballots. In later sections we show why average number of ballots is not the only metric to optimize, and how large round sizes can be beneficial from the perspective of other important metrics. 

\begin{figure}[h!]
\includegraphics[width=.5\textwidth]{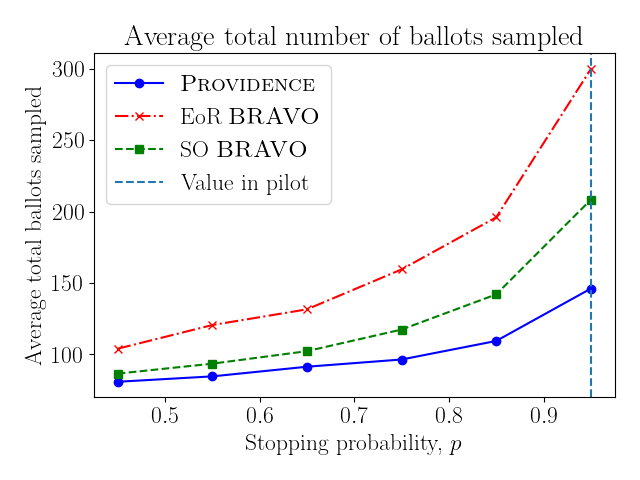}
\caption{The total number of ballots sampled on average as a function of $p$, the conditional stopping probability used to select each round size. We use the same contest parameters and risk limit as the Rhode Island pilot.}
\label{fig:pilot_sims}
\end{figure}

For this pilot audit, extensive planning of the round schedule was not necessary because the margin was large enough that relatively few ballots were needed to achieve the high probability of stopping. In Section~\ref{sec:workload} we consider a larger state-wide contest in Virginia, where selecting the round schedule has more significant implications. Virginia also currently uses ballot polling RLAs, whereas Rhode Island primarily uses batch comparison RLAs. Some of the ideas introduced in Section~\ref{sec:workload} provide a context for this pilot case as well.

For the sake of analysis, the selection order of the ballots sampled during the pilot was also recorded. Figure~\ref{fig:pilot_sequence} shows the cumulative tally of winner ballots after each new ballot in the selection order is added to the sample. We observe two interesting phenomena in this particular sample's selection order. 
\begin{description}
\item First, an SO \BRAVO audit of this sample stops because the \BRAVO condition is met when the sample (orange line) surpasses the minimum number of winner ballots (blue line) earlier in the sample.\footnote{Such cases also provide insight into how \Providence is a tighter test in expectation because SO \BRAVO ignores information from the rest of the sample after the \BRAVO condition is met at some point earlier in the selection order.} EoR \BRAVO, however, does not stop. It might be difficult to explain to the public why SO \BRAVO stops in more extreme cases like this, where the condition is met early in the sample, but poor evidence for the alternative hypothesis in the rest of the sample is ignored. 
\item Second, only $5$ of the first $11$ ballots were for the announced winner. A first round of size $11$ would have resulted in a smaller average total number of ballots drawn, but would have provided a misleading sample (suggesting that the winner was incorrectly reported) due to a too-small sample size. 
\end{description}
Both these ideas are addressed more thoroughly in Section~\ref{sec:workload}.

\begin{figure}
\includegraphics[width=.5\textwidth]{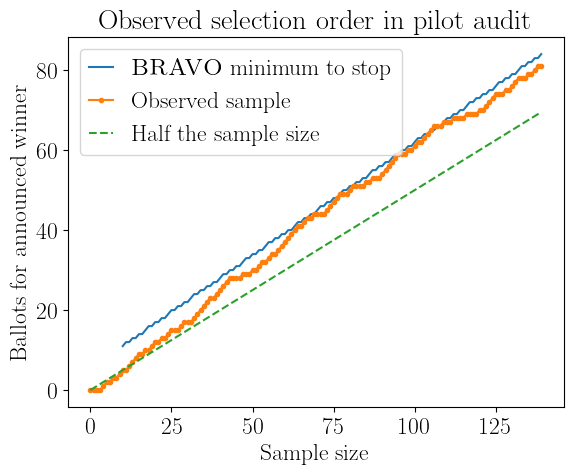}
\caption{For each sample size from $1$ to $140$, the intermediate cumulative sum of ballots for the announced winner found in the sample is shown.}
\label{fig:pilot_sequence}
\end{figure}

\section{Audit workload}
\label{sec:workload}
Some election audits have benefited from a one-and-done approach: draw a large sample with high probability of stopping in the first round and usually avoid a second round altogether. This is appealing for two reasons. Firstly, rounds have some overhead in both time and effort. Thus the time and person-hours of an audit grows not just with the number of ballots sampled but also with the number of rounds. Secondly, smaller first round sizes are not large enough to accurately capture the distribution of votes. There is a higher probability that the true winner has fewer votes in the audit sample than some other candidate. On the other hand, a one-and-done audit may draw more ballots than are necessary; a more efficient round schedule could require less effort and time pre-certification. To evaluate the quality of various round schedules, we construct a simple workload model. Using this model we show how optimal round schedules can be chosen. We provide software that can be used by election officials to choose round schedules based on estimates of the model parameters like maximum allowed probability of a misleading audit sample.

As an example, we consider the US Presidential contest in the 2016 Virginia statewide general election. This contest had a margin of $0.053$ between the two candidates with the most votes.
Analytical approximation of the expected audit behavior (quantities like expected total number of ballots sampled or total number of rounds) is challenging because the number of possible sequences of samples grows exponentially with the number of rounds. 
Therefore we use the typical approach of simulations, again with risk limit $0.1$.

We simulate audits considering each candidate with a column in the results available at the Virginia Department of Elections website, including irrelevant ballots.
We consider a simple round schedule, in which each round is selected to give the same probability of stopping, $p$. That is, if the audit does not stop in the first round, we select a second round size which, given the sample drawn in the first round, will again have a probability of stopping $p$ in the second round. Note that since there are multiple candidates, we compute the minimum round size to achieve stopping probability $p$ for each pairwise contest between the winner and one of the losers, and we then select the largest such minimum round size and scale it up according to the proportion of the total ballots that are relevant to that pairwise contest. For this round schedule scheme, a one-and-done audit is achieved by choosing large $p$, say $p=.9$ or $p=.95$. We run $10^4$ trial audits for each value of $p$, assuming the reported results are correct\footnote{For this particular round schedule scheme, computing the expected number of rounds is straightforward analytically, but the expected number of ballots is still difficult, and so we use simulations.}. 

Note that simulations of audits of tied elections are not necessary, as all the audits we are considering are risk-limiting and hence we already know the performance to expect when auditing a tied election, even one not reported as such. 

\subsection{Person-hours}

\subsubsection{Average total ballots.} 
The simplest workload models are a function of just the total number of a ballots sampled.\footnote{Sometimes total \emph{distinct} ballots sampled is used, but for the margins we use in our examples in this section, the difference between total distinct ballots and total ballots is insignificant\cite{arxiv_athena}. It is straightforward to modify the model we discuss here to account for total distinct ballots.} Figure~\ref{fig:avg_bals} shows the average total number of ballots sampled as a function of $p$.
\begin{figure}[h!]
\includegraphics[width=.5\textwidth]{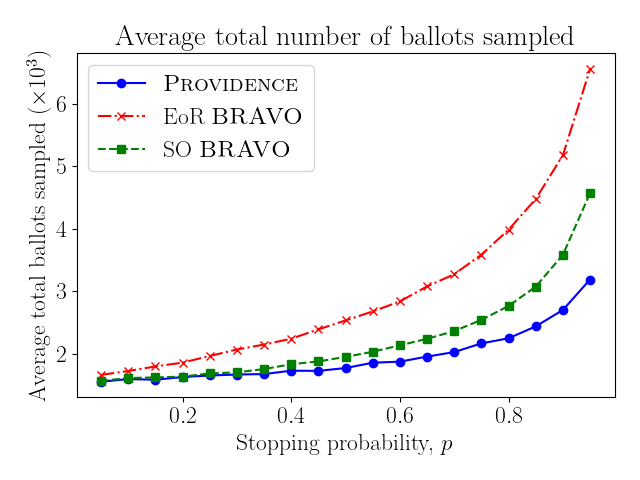}
\caption{The total number of ballots sampled on average, as a function of $p$, the conditional stopping probability used to select each round size, for ballot polling audits of the 2020 US Presidential election in the state of Virginia.}
\label{fig:avg_bals}
\end{figure}

Figure~\ref{fig:avg_bals_ratio} provides the same number as a fraction of the \Providence values.
It is straightforward to show that \Providence and both forms of \BRAVO collapse to the same test when each round corresponds to a single ballot. Figures~\ref{fig:avg_bals} and \ref{fig:avg_bals_ratio} show that for larger stopping probabilities $p$ (i.e. larger rounds), \Providence requires fewer ballots on average. In particular, the savings of \Providence become larger as $p$ increases; for $p=0.95$, EoR \BRAVO and SO \BRAVO require more than $2$ and $1.4$ times as many ballots as \Providence respectively. 
\begin{figure}[h!]
\includegraphics[width=.5\textwidth]{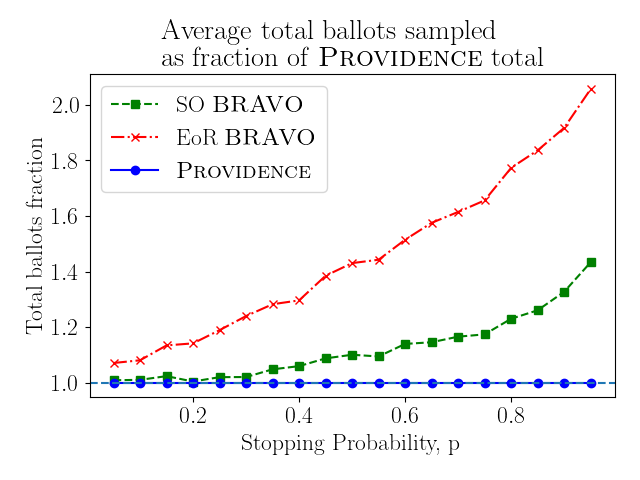}
\caption{The total number of ballots sampled on average, as a fraction of those sampled by \Providence, as a function of $p$, the conditional stopping probability used to select each round size, for ballot polling audits of the 2020 US Presidential election in the state of Virginia.}
\label{fig:avg_bals_ratio}
\end{figure}

\subsubsection{Round overhead.} 
It is clear that average number of ballots alone is an inadequate workload measure. 
(Consider a state conducting its audit by selecting a single ballot at random, 
notifying just the county where the ballot is located, and then waiting to hear back for the manual interpretation of the ballot before moving on to the next one. 
This of course is inefficient and is why audits are actually performed in rounds.)

In a US state-wide RLA, the state organizes the audit by determining the random sample and communicating with the counties, but election officials at the county level physically sample and inspect the ballots after drawing them from secure storage boxes stored in county locations. 
Therefore each audit round requires some number of person-hours for set up and communication between state and county. This overhead for a round includes choosing the round size, generating the random sample, and communicating that random sample to the counties, as well as the communication of the results back to the state afterwards. 

Consequently, we now consider a model with a constant per-ballot workload $w_b$ and a constant per-round workload $w_r$.
So for an audit with expected number of ballots $E_{b}$ and expected number of rounds $E_{r}$, we estimate that the workload $W$ of the audit is
\begin{equation}
W(E_b,E_r) = E_b w_b + E_r w_r + C
\label{eq:round_workload}
\end{equation}

Note there is also some constant overhead of workload for the whole audit, namely $C$ in Equation~\ref{eq:round_workload}, which we take to be zero in our examples but could be used by election officials to represent, for example, the effort of constructing a ballot manifest.
For simplicity, (and without loss of generality), we measure in multiples of the per ballot workload; that is, we assume it is one unit, $w_b=1$. A per round workload of $w_r=x$ corresponds to a per round workload which is $x$ times the per ballot workload. We use $w_r=1000$ as a conservative example. 
That is, we set the overhead of a round equal to the workload of sampling $1000$ ballots. Based on available data\cite{RI-report}, the time retrieving and analyzing each individual ballot is on the order of $75$ seconds which means that $w_r=1000$ is equivalent to roughly $20$ person-hours of workload. This corresponds to about $15$ minutes being spent, on average, per round in each of the $133$ counties of Virginia, a clearly conservative workload estimate. 

As shown in Figure~\ref{fig:with_round_workload}, average workloads first reduce as stopping probability increases; this is likely due to a decrease in the number of rounds. After hitting a sweet spot, average workloads again increase with stopping probability; this time, likely because the average number of rounds does not increase much and the cost changes because of number of ballots drawn, which increases with round size. \Providence achieves the lowest minimum average workload at roughly $p=0.7$ for our example choice of $w_r=1000$.

\begin{figure}[h!]
\includegraphics[width=.5\textwidth]{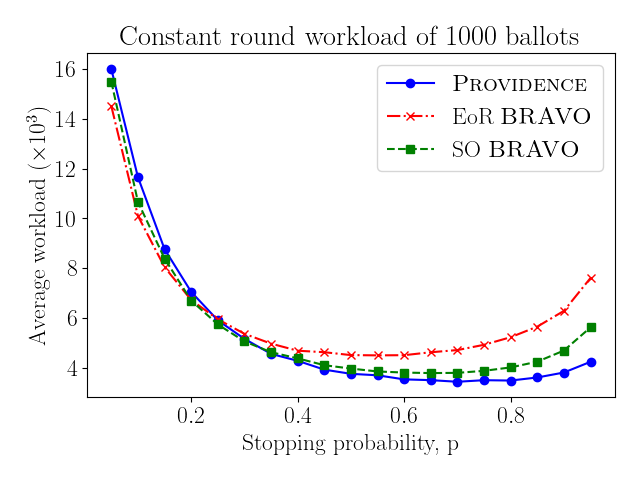}
\caption{For workload parameters $w_b=1$ and $w_r=1000$, this plot shows the expected workload for various round schedule parameters $p$. Expected workload is found using Equation~\ref{eq:round_workload} and the average number of ballots and rounds in our simulations as the expected number of ballots and rounds.}
\label{fig:with_round_workload}
\end{figure}

Importantly, this gives us a way to estimate the minimum expected workload, as well as which round schedule value $p$ achieves it, for arbitrary round workload. For each round workload $w_r$, we produce a dataset analagous to that of Figure~\ref{fig:with_round_workload} and then find the minimum average workload achieved for each of the audits and its corresponding stopping probability $p$. 

Figure~\ref{fig:optimal_workloads} shows the optimal achievable workload for a wide range of per round workloads. For very low round workloads, the workload function approaches just the total number of ballots, and so workload is minimized by minimizing the number of ballots drawn, which corresponds to small round sizes, and we would expect all three audits to behave similarly, as ballot-by-ballot audits, with the smallest workload. On the other hand, for extremely large values of round workload, the average number of ballots has little impact on the workload function, and so the three audits again have similar values, all corresponding to large round sizes in order to minimize the number of rounds.  We know that there is variation in the number of ballots used by each type of audit for large round sizes (a factor of two for $p=0.9$), but these values would be small in comparison to $w_r$. We observe this behaviour in Figure~\ref{fig:optimal_workloads} for extremely small and large workload values. For more reasonable values of the round workload $w_r$, SO \BRAVO and EoR \BRAVO achieve minimum workload roughly $1.1$ and $1.3$ times greater than that of \Providence.
\begin{figure}[h!]
\includegraphics[width=.5\textwidth]{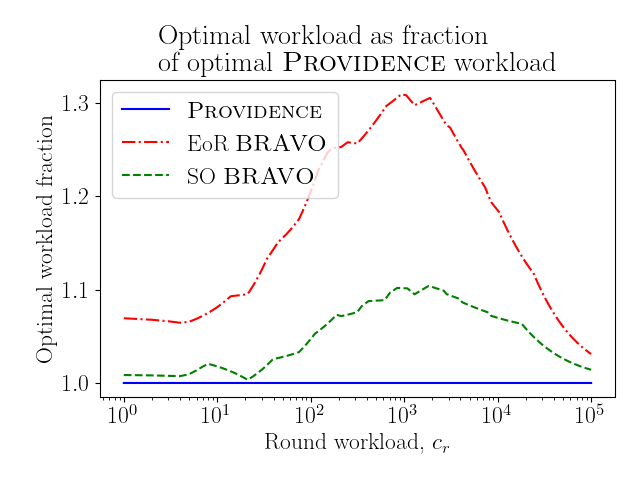}
\caption{For varying round workload $w_r$, the optimal average workload achievable by each audit, as a fraction of the \Providence values.}
\label{fig:optimal_workloads}
\end{figure}

Figure~\ref{fig:optimal_ps} shows the corresponding round schedule parameters $p$ that achieve these minimal workloads. As expected, an overhead for each round means that larger round sizes are needed to achieve an optimal audit, and so for all three audits $p$ increases as a function of $w_r$. Notice that \Providence is generally above and to the left of SO \BRAVO, and SO \BRAVO is generally above and to the left of EoR \BRAVO. This relationship reflects the fact that for the same round workload, \Providence can get away with a larger stopping probability because it requires fewer ballots.
\begin{figure}[h!]
\includegraphics[width=.5\textwidth]{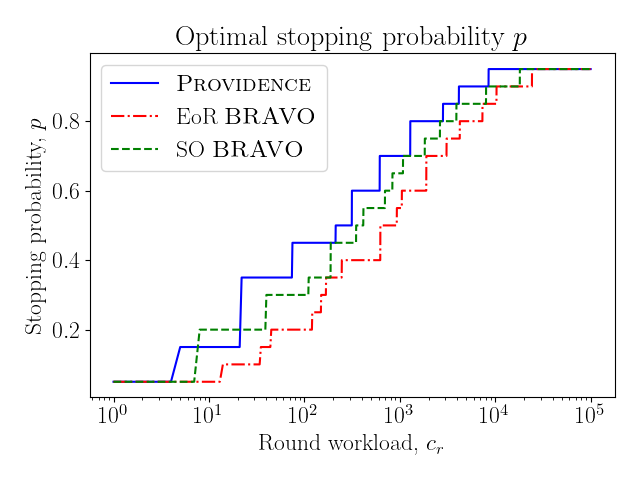}
\caption{The optimal (workload-minimizing) stopping probability $p$ for varying workload model parameters $w_r$. (Note that the steps in this function are a consequence of our subsampling the workload function. That is, the workload-minimizing value of $p$ for each $w_r$ is only allowed to take on values at increments of $0.05$.)}
\label{fig:optimal_ps}
\end{figure}

\subsubsection{Precinct overhead.} For a more complete model, we can also introduce container-level workload. If. around requires multiple ballots from a single container, the container need only be unsealed once. Based on a Rhode Island pilot RLA report\cite{RI-report}, this may mean that a ballot from a new container requires roughly twice the time as a ballot from an already-opened container. Typically available election results give per-precinct granularity of vote tallies, rather than individual container information. In Virginia, however, most precincts have a single ballot scanner whose one box has sufficient capacity for all the ballots cast in that precinct anyways, and so we model the per-container workload as a per-precinct workload, $w_p$. In this model, the workload estimate incurs an additional workload of $w_p$ every time a precinct is sampled from for the first time in a round. That is, let $E_{pi}$ be the expected number of distinct precincts sampled from in round $i$, and let $E_p=\sum_i E_{pi}$. Then the new model is
\begin{equation}
W(E_b, E_r, E_p) = E_b w_b + E_r w_r + E_p w_p + C
\label{eq:round_and_precinct_workload}
\end{equation}

We can again explore the minimum achievable workloads under this model, as shown in Figure~\ref{fig:optimal_workload_precinct_workload_ratio}.

\begin{figure}
\includegraphics[width=.5\textwidth]{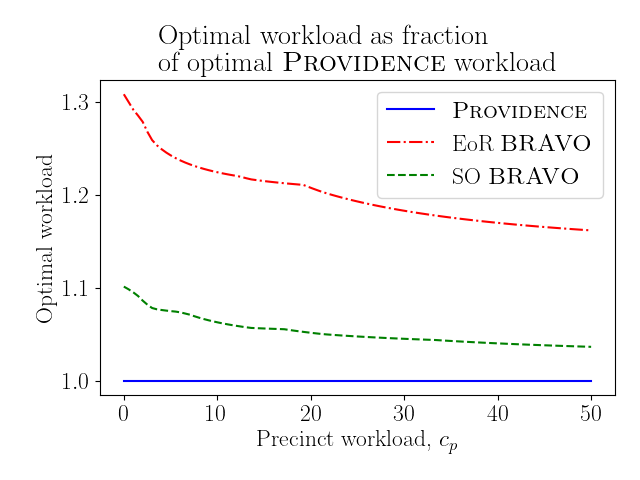}
\caption{Optimal average workload using the workload Equation~\ref{eq:round_and_precinct_workload} for varying $w_p$, given as a fraction of the value for \Providence. Similar to Figure~\ref{fig:optimal_workloads}, we show a generous range of values for the workload variable, $c_p$ in this case. If the time for a single ballot is $75$ seconds, then $c_p=50$ corresponds to over an hour of extra time to sample a ballot from a new container.}
\label{fig:optimal_workload_precinct_workload_ratio}
\end{figure}

\subsection{Real time}
Given tight certification deadlines\footnote{Virginia recently passed legislation requiring pre-certification RLAs.}, the total real time to conduct the RLA is also an important factor to consider when planning audits.
Because each county can sample ballots for the same round concurrently, the total real time for a round depends only on the slowest county. 
In Virginia, Fairfax County typically has the most votes cast by a significant difference; in the contest we consider, Fairfax County had ~551 thousand votes cast, more than double the ~203 thousand of second-highest Virginia Beach City.
Consequently, we model the expected total real time $T$ of an audit using just the largest county, and we define analagous variables for the expected values in just the largest county.
For the largest county, let the expected total ballots sampled be $\bar E_b$, the expected number of rounds $\bar E_r$, and the expected number of distinct precinct samples summed over all rounds be $\bar E_c$.
Similarly, we use real time per-ballot, per-round, and per-precinct workload variables, $t_b$, $t_r$, and $t_p$. So the real time of the audit is estimated by
\begin{equation}
T(\bar E_b, \bar E_r, \bar E_p ) = \bar E_b t_b + \bar E_r t_r + \bar E_p t_p + C
\label{eq:real_time}
\end{equation}

As before, we can use our simulations to estimate $\bar E_b$, $\bar E_r$, and $\bar E_p$ using the corresponding averages over the trials. 
Available data to estimate values for $t_b$, $t_r$, and $t_p$ is limited, and so we take as an example the values $t_b=75$ seconds, $t_r=3$ hours, and $t_p=75$ seconds.\footnote{The value $t_b=75$ seconds corresponds to a serial retrieval and interpretation of the ballots based on the \cite{RI-report} timing, $t_p=75$ seconds corresponds to the approximate doubling in time for new-box ballots as reported in \cite{RI-report} in the ballot-level comparison timing data, and $t_r=3$ hours is just a guess at an approximate order for this variable.} In practice, election officials could use our software and their own estimates of these values to explore choices for round schedules. Figure~\ref{fig:real_time} shows how the estimated real time for these values differs as a function of $p$. It should be noted that real values of $t_b$, $t_r$, and $t_p$ will vary greatly based on the number of parallel teams retrieving and checking ballots, the distribution of ballots and containers both in number and physical space, and other factors. We provide Figure~\ref{fig:real_time} only as an example of the general shape and behavior of this function. Use of this optimal scheduling tool would depend on parameter estimates tailored to each case.

\begin{figure}
\includegraphics[width=.5\textwidth]{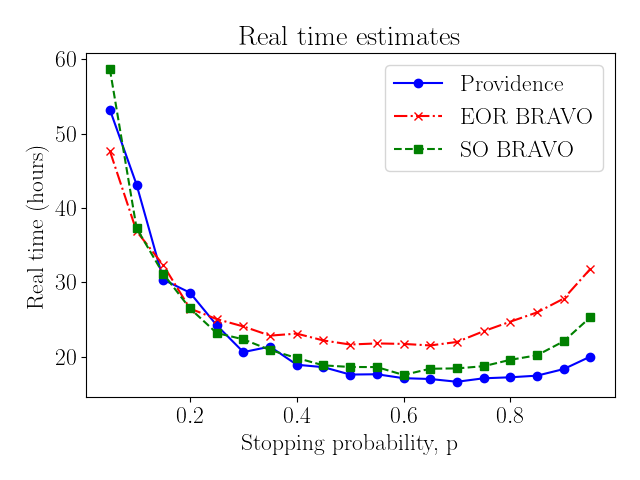}
\caption{The real time as estimated by Equation~\ref{eq:real_time} for varying $p$ with expected values as estimated by our simulations.}
\label{fig:real_time}
\end{figure}

\subsection{Misleading samples}

Unfortunately, efficiency alone is not sufficient for planning audits. In the US today, election officials have a legitimate need to include personal safety as a consideration.
In a random sample, a true loser may receive more votes than the true winner. This happens more often when the sample sizes are small, like for a hypothetical first round size of $11$ in the pilot audit, as seen in Figure~\ref{fig:pilot_sequence}.
In the abstract, a misleading sample in an early round is dealt with by drawing more ballots (moving on to another round), but in practice the implications of this approach may be dangerous.

Imagine that Alice beats Bob in an election contest both truly and in the reported results, but Bob's supporters are insistent he really won. When election officials carry out the RLA, they choose a small first round size in the hopes of achieving an efficient audit by getting to stop sooner (and drawing fewer ballots on average). After the first round, by chance, there are more votes for Bob than for Alice in the sample. Bob's supporters celebrate their victory that the audit has in fact revealed that Bob really won, but the election officials have to explain that they are moving on to a second round. After the second round, there are more votes in the sample for Alice and sufficiently many that the risk limit is met and the audit now ends confirming the announced result that Alice won. This is an undesirable situation, as it can appear to Alice's supporters that election officials are simply drawing ballots till a chosen outcome is obtained. 

We introduce the notion of a \emph{misleading sample}, any cumulative sample which, assuming the announced outcome is correct, contains more ballots for a loser than for the winner.
We can again use our simulations to gain insight into the frequency of \emph{misleading samples}.
For each stopping probability $p$, Figure~\ref{fig:misleading} gives the proportion of simulated audits that had a \emph{misleading sample} at any point. 
Notably, this proportion is as high as 1 in 5 for the smaller stopping probability round schedules.
Accordingly, we introduce a new parameter to our audit-planning tool, the maximum acceptable probability that the audit is misleading, the \emph{misleading limit}.

In Figure~\ref{fig:misleading}, horizontal lines are included to show \emph{misleading limits} of $0.1$, $0.01$, and $0.001$.
To achieve a probability of a misleading sample of at most $0.1$, a round schedule with at least roughly $p=.3$ is needed.
To achieve a probability of misleading of roughly $0.01$, a round schedule with $p=0.8$ is needed, and to achieve a probability of misleading of roughly $0.001$, a round schedule with $p=0.95$ is needed.
It is not unreasonable to think that election officials might choose a \emph{misleading limit} of $0.01$, or smaller, given the state of public perception of election security in the US and the associated threats of violence.
Consequently, the desired \emph{misleading limit} may be a deciding constraint in the choice of round schedule. 

\begin{figure}
\includegraphics[width=.5\textwidth]{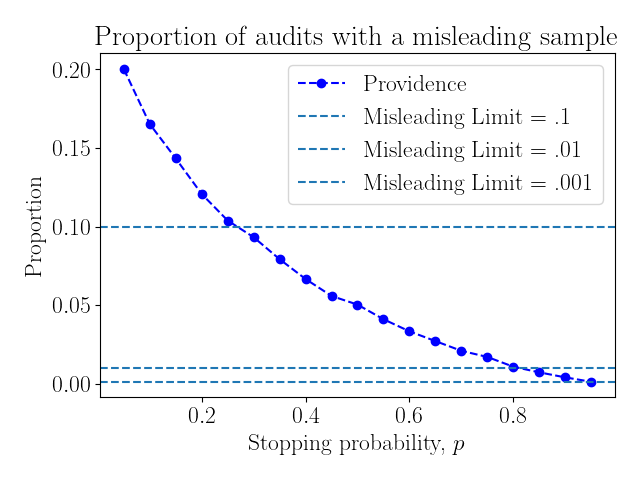}
\caption{The proportion of simulated \Providence audits for the Virginia contest parameters that had a \emph{misleading sample} in any round.}
\label{fig:misleading}
\end{figure}

We observe a similar behavior in our simulations of audits on the contest from the pilot audit. Figure~\ref{fig:pilot_misleading} shows the proportion of the pilot simulations which contained a \emph{misleading sample} in any round. Despite the large difference in margin ($\sim 0.05$ in Virginia and $\sim 0.25$ in the pilot) we still observe that a \emph{misleading limit} of $0.01$ is first achieved at roughly $p=0.8$ and $0.001$ at $p=0.95$.

\begin{figure}
\includegraphics[width=.5\textwidth]{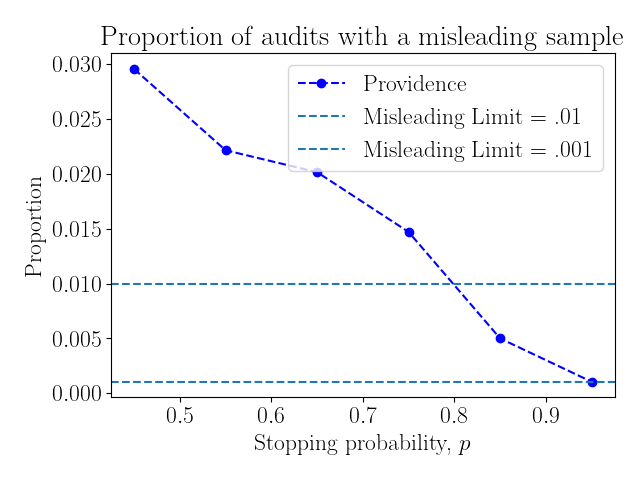}
\caption{The proportion of simulated \Providence audits for the pilot audit parameters that had a \emph{misleading sample} in any round.}
\label{fig:pilot_misleading}
\end{figure}

If election officials wish to enforce a \emph{misleading limit} for all the rounds, our simulation analysis could help. On the other hand, for a given round, it is straightforward to compute analytically the probability that a loser has more votes than the winner in the sample. Table~\ref{tab:misleading} shows for various margins the minimum first round size $n$ that guarantees a probability of a \emph{misleading sample} at most $M\in\{0.1,0.01,0.001\}$. For all values of $M$ and all margins, \Providence achieves a higher probability of stopping than either EoR \BRAVO or SO \BRAVO. 
    As seen in the Table~\ref{tab:misleading}, to enforce $M=0.01$ requires minimum round sizes with at least roughly a $0.8$ probability of stopping in the first round. Even if the most efficient audit schedule (by either workload or real time measures) would use a lower stopping probability $p$ to choose the first round size, the election officials may opt to use this constraint on the probability of a \emph{misleading sample} as the deciding factor in planning their audits.

\subsubsection{Misleading SO \BRAVO sequences.} As we consider the idea of misleading samples, it is noteworthy that SO \BRAVO suffers from a different and unique type of misleading result. 

After drawing a cumulative $n>1$ ballots in a round, some number $k$ of them are votes for the announced winner. There are $\binom{n}{k}$ possible sequences of ballots which can lead to such a sample. Given a value of $k$, however, the particular sequence of the sample that led to that value of $k$ contains no additional information about whether the sample is more likely under the alternative or null hypotheses. That is to say, $\Pr[K=k|H_a]$ and $\Pr[K=k|H_0]$ have the same value regardless of the sequence.
Despite this, the SO \BRAVO RLA stopping condition is not just a function of $n$ and $k$ but also a function of the sequence, the selection order. In particular, if the sequence of ballots is such that the standard \BRAVO stopping condition was met for some $n'<n$ and corresponding $k'<n$, the audit will stop, even if by the end of the sequence the values $k$ and $n$ no longer meet the \BRAVO condition. We refer to such sequences which stop under SO \BRAVO, but not under EoR \BRAVO, as \emph{misleading sequences}. To be clear, this is not a mathematical issue; stopping in such cases is still a correct application of Wald's SPRT result\cite{wald}. The misleading nature of such stoppages is the note we are making. This is another case where election officials might have difficulty explaining the misleading situation to the public.

Recall from Section~\ref{sec:pilot} that the pilot \Providence RLA performed in Providence, Rhode Island had an SO \BRAVO \emph{misleading sequence}. In particular, the audit passed with an SO \BRAVO risk measure of $0.0541$ but the final cumulative tally of the sample gives a \BRAVO risk measure of $0.366$.

It is easy to use our simulations to see how often SO \BRAVO \emph{misleading sequences} occur by checking whether the final cumulative sample of each SO \BRAVO trial meets the EoR \BRAVO stopping condition and counting those which do not. Figure~\ref{so_misleading} shows the proportion of simulated SO \BRAVO audits that stopped with a \emph{misleading sequence}. Unlike the more general \emph{misleading sample} discussed so far, these \emph{misleading sequences} are unique to SO \BRAVO audits, and Figure~\ref{so_misleading} only shows the proportion of audits that stopped with a \emph{misleading sequence}; additional SO \BRAVO audits also contained \emph{misleading samples}.

\begin{figure}
\includegraphics[width=.5\textwidth]{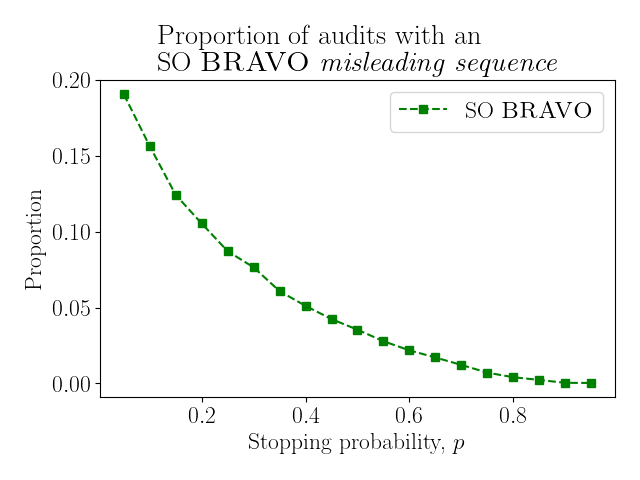}
\caption{The proportion of sequences that are misleading sequences in the SO \BRAVO audit as a function of $p$.}
\label{so_misleading}
\end{figure}

\section{Conclusion}
\label{sec:conc}
A rigorous tabulation audit is an important part of a secure election. We present \Providence and demonstrate that it is as efficient as \Minerva and as flexible as \BRAVO. We present proofs and simulation results to verify the claimed properties of \Providence, and we provide an open source implementation of the stopping condition and useful related functionality for planning audits. We define the constraint of an acceptable probability of a misleading audit sample, and describe its importance to the planning process.

\begin{table}[h!]
\center
\begin{tabular}{ |l|l|r|c|c|c| }
\hline
$M$ & margin & $n$ & Prov & SO & EoR \\
\hline
0.1&0.25&25&0.221&0.152&0.115\\
&0.15&73&0.202&0.186&0.141\\
&0.05&657&0.227&0.192&0.127\\
&0.03&1825&0.246&0.194&0.124\\
&0.01&16423&0.246&0.196&0.124\\
\hline
0.01&0.25&85&0.792&0.707&0.559\\
&0.15&239&0.817&0.712&0.549\\
&0.05&2163&0.817&0.721&0.569\\
&0.03&6011&0.824&0.723&0.573\\
&0.01&54117&0.824&0.724&0.57\\
\hline
0.001&0.25&149&0.962&0.889&0.783\\
&0.15&421&0.958&0.894&0.801\\
&0.05&3815&0.96&0.896&0.785\\
&0.03&10607&0.961&0.897&0.787\\
&0.01&95491&0.962&0.897&0.787\\
\hline
\end{tabular}
\caption{For various margins, this table gives the minimum first round size $n$ to achieve at most a probability $M$ of a \emph{misleading sample} in the first round. The corresponding stopping probabilities of \Providence, SO \BRAVO, and EoR \BRAVO are given for each value of $n$.}
\label{tab:misleading}
\end{table}

\section{Availability}
\label{sec:avail}
\Providence is implemented in the open source R2B2 software library for \R and \B audits \cite{r2b2} 
We provide software to test stopping conditions, find round sizes to achieve a given probability of stopping, and find round sizes that have acceptable probabilities of a \emph{misleading sample}. The software also includes the simulator for all these audits and the functionality to perform the workload and real time analysis we present in this paper.

\section{Acknowledgements}
\label{sec:ack}
The authors are grateful to the Rhode Island Board of Elections for conducting a pilot \Providence RLA, and to Georgina Cannan, Liz Howard, Mark Lindeman, and John Marion for their support of the pilot.  The authors thank Audrey Malagon for useful information on audits.

\bibliographystyle{plain}
\bibliography{audits.bib}

\begin{thebibliography}{10}

\bibitem{bernhard-diss}
Matthew Bernhard.
\newblock {\em Election Security Is Harder Than You Think}.
\newblock PhD thesis, University of Michigan, 2020.

\bibitem{bernhard-sok}
Matthew Bernhard.
\newblock Risk-limiting {A}udits: A practical systematization of knowledge.
\newblock In {\em In Proceedings, Seventh International Joint Conference on
  Electronic Voting (E-Vote-ID'21), October 2021}, 2021.

\bibitem{blom_IRV}
Michelle~L. Blom, Peter~J. Stuckey, and Vanessa~J. Teague.
\newblock Ballot-polling risk limiting audits for {IRV} elections.
\newblock In Robert Krimmer, Melanie Volkamer, V{\'{e}}ronique Cortier, Rajeev
  Gor{\'{e}}, Manik Hapsara, Uwe Serd{\"{u}}lt, and David Duenas{-}Cid,
  editors, {\em Electronic Voting - Third International Joint Conference,
  E-Vote-ID 2018, Bregenz, Austria, October 2-5, 2018, Proceedings}, volume
  11143 of {\em Lecture Notes in Computer Science}, pages 17--34. Springer,
  2018.

\bibitem{simulations}
Oliver Broadrick, Sarah Morin, Grant McClearn, Neal McBurnett, Poorvi~L. Vora,
  and Filip Zag\'{o}rski.
\newblock Simulations of ballot polling risk-limiting audits.
\newblock In {\em Seventh Workshop on Advances in Secure Electronic Voting, in
  Association with Financial Crypto}, 2022.

\bibitem{RI-report}
Common Cause, VerifiedVoting, and Brennan Center.
\newblock Pilot implementation study of risk-limiting audit methods in the
  state of rhode island.
\newblock
  \url{https://www.brennancenter.org/sites/default/files/2019-09/Report-RI-Design-FINAL-WEB4.pdf}.

\bibitem{DBLP:conf/evoteid/HuangRSTV20}
Zhuoqun Huang, Ronald~L. Rivest, Philip~B. Stark, Vanessa~J. Teague, and Damjan
  Vukcevic.
\newblock A unified evaluation of two-candidate ballot-polling election
  auditing methods.
\newblock In Robert Krimmer, Melanie Volkamer, Bernhard Beckert, Ralf
  K{\"{u}}sters, Oksana Kulyk, David Duenas{-}Cid, and Mikhel Solvak, editors,
  {\em Electronic Voting - 5th International Joint Conference, E-Vote-ID 2020,
  Bregenz, Austria, October 6-9, 2020, Proceedings}, volume 12455 of {\em
  Lecture Notes in Computer Science}, pages 112--128. Springer, 2020.

\bibitem{RLA}
Mark Lindeman and Philip~B Stark.
\newblock A gentle introduction to risk-limiting audits.
\newblock {\em IEEE Security \& Privacy}, 10(5):42--49, 2012.

\bibitem{bravo}
Mark Lindeman, Philip~B Stark, and Vincent~S Yates.
\newblock {BRAVO}: Ballot-polling risk-limiting audits to verify outcomes.
\newblock In {\em EVT/WOTE}, 2012.

\bibitem{simulations_house}
Katherine McLaughlin and Philip~B. Stark.
\newblock Simulations of risk-limiting audit techniques and the effects of
  reducing batch size on the 2008 {California House of Representatives}
  elections.
\newblock NSF report, 2010.

\bibitem{mclaughlin_thesis}
Katherine McLaughlin and Philip~B. Stark.
\newblock Workload estimates for risk-limiting audits of large contests.
\newblock Honors Thesis, University of California, Berkeley, 2011.

\bibitem{r2b2}
Sarah Morin and Grant McClearn.
\newblock The {R2B2} ({R}ound-by-{R}ound, {B}allot-by-{B}allot) library,
  https://github.com/gwexploratoryaudits/r2b2.

\bibitem{bernoulli-ballot-polling}
Kellie Ottoboni, Matthew Bernhard, J.~Alex Halderman, Ronald~L Rivest, and
  Philip~B. Stark.
\newblock Bernoulli ballot polling: A manifest improvement for risk-limiting
  audits.
\newblock {\em International Conference on Financial Cryptography and Data
  Security}, pages 226--241, 2019.

\bibitem{rivest2008notion}
Ronald~L Rivest.
\newblock On the notion of "software independence" in voting systems.
\newblock {\em Philosophical Transactions of the Royal Society A: Mathematical,
  Physical and Engineering Sciences}, 366(1881):3759--3767, 2008.

\bibitem{SI-Wack}
Ronald~L. Rivest and John~P. Wack.
\newblock On the notion of ``software independence'' in voting systems.
\newblock Prepared for the TGDC, and posted by NIST at the given url.

\bibitem{dice}
Philip~B. Stark.
\newblock Simulating a ballot-polling audit with cards and dice.
\newblock In {\em Multidisciplinary Conference on Election Auditing, MIT},
  december 2018.

\bibitem{evidence-based}
Philip~B. Stark and David~A. Wagner.
\newblock Evidence-based elections.
\newblock {\em {IEEE} Secur. Priv.}, 10(5):33--41, 2012.

\bibitem{Bayesian-RLA}
Poorvi~L. Vora.
\newblock Risk-limiting {B}ayesian polling audits for two candidate elections.
\newblock {\em CoRR}, abs/1902.00999, 2019.

\bibitem{vv_audits}
Verified Voting.
\newblock Audit law database, https://verifiedvoting.org/auditlaws/.

\bibitem{arlo}
VotingWorks.
\newblock Arlo, https://voting.works/risk-limiting-audits/.

\bibitem{wald}
Abraham Wald.
\newblock Sequential tests of statistical hypotheses.
\newblock {\em The Annals of Mathematical Statistics}, 16(2):117--186, 1945.

\bibitem{arxiv_athena}
Filip Zag{\'{o}}rski, Grant McClearn, Sarah Morin, Neal McBurnett, and
  Poorvi~L. Vora.
\newblock The {Athena} class of risk-limiting ballot polling audits.
\newblock {\em CoRR}, abs/2008.02315, 2020.

\bibitem{usenix_minerva}
Filip Zag{\'o}rski, Grant McClearn, Sarah Morin, Neal McBurnett, and Poorvi~L.
  Vora.
\newblock Minerva{\textendash} an efficient risk-limiting ballot polling audit.
\newblock In {\em 30th {USENIX} Security Symposium ({USENIX} Security 21)},
  pages 3059--3076. {USENIX} Association, August 2021.

\end{thebibliography}


\appendix

\section{Proofs}
\label{sec:proofs}
\textbf{Lemma \ref{lem:efficiency}.}
For any risk-limit $\alpha \in (0, 1)$, for any margin
and for any round schedule $[n_1, \ldots, n_j]$, 
the \Providence RLA is more efficient than EoR \BRAVO.

\begin{proof}
Let $[n_1, \ldots, n_j]$ be a round schedule, and assume that an EoR \BRAVO audit stops in round $j$, after observing $k_1, \ldots, k_j$ ballots for the announced winner in each round respectively.
That is, the EoR \BRAVO stopping condition is true:
$$\sigma(k_j,p_a,p_0,n_j) \ge \frac{1}{\alpha}.$$
To see the \Providence stopping condition is fulfilled, we rewrite as 



\[
 \frac{1}{\alpha} \leq \sigma(k_{j}, p_a, p_0, n_{j}) 
\]
\[
 = \sigma(k_{j-1}, p_a, p_0, n_{j-1}) \cdot \sigma(k_j - k_{j-1}, p_a, p_0, n_j - n_{j-1})  
\]
\[
 \leq^{(*)} \sigma(k_{j-1}, p_a, p_0, n_{j-1}) \cdot \tau_1(k_j - k_{j-1}, p_a, p_0, n_j - n_{j-1}) 
\]
\[
 = \omega_r(k_j, k_{j-1}, p_a, p_0, n_j, n_{j-1}).
\]

Where inequality $(*)$ follows from \cite[Theorem 6]{arxiv_athena}. Note that we apply this result on $\tau_j$ for just $j=1$.

%
%
%
\end{proof}
\begin{lemma}
    \label{lemma:sigma_increasing}
For $0<p_0< p_a< 1$ and $n>0$, the ratio $\sigma(k,p_a,p_0,n)$ is strictly increasing as a function of $k$ for $0\le k\le n$.
\end{lemma}
\begin{proof}
See \cite[Lemma 4]{usenix_minerva}. 
\end{proof}

\begin{lemma}
    \label{lemma:frac_sums_increasing}
    Given a monotone increasing sequence: $\frac{a_1}{b_1}, \frac{a_2}{b_2}, \ldots, \frac{a_n}{b_n}$, for $a_i, b_i > 0$, the sequence:
    $$z_i = \frac{\sum_{j=i}^n a_j}{\sum_{j=i}^n b_j}$$
    is also monotone increasing.
\end{lemma}

\begin{proof}
See \cite[Lemma 2]{usenix_minerva}. 
\end{proof}

\begin{lemma}
    \label{lemma:tau1_increasing}
For $0<p_0< p_a< 1$ and $n>0$, the ratio $\tau_1(k,p_a,p_0,n)$ is strictly increasing as a function of $k$ for $0\le k\le n$.
\end{lemma}
\begin{proof}
    Apply Lemmas \ref{lemma:sigma_increasing}-\ref{lemma:frac_sums_increasing}.
\end{proof}
\begin{lemma}
    \label{lemma:imin-exists}
    Given a strictly monotone increasing sequence: $x_1, x_2, \ldots x_n $ and some constant $A$,
    $$A \le x_i \Leftrightarrow \exists i_{min} \le i ~\text{s.t.}~   x_{i_{min} -1} < A \le x_{i_{min}} \le x_{i},$$
    unless $A\le x_1$, in which case $i_{min} =1 $.
\end{lemma}
\begin{proof}
    Evident.
\end{proof}

\begin{lemma}
    \label{lemma:minerva2_kmin_exists}
    For $\mathcal{A}=(\alpha,p_a, p_0,k_{j-1},n_{j-1},n_j)$-\Providence, there exists\\ a 
    $k^{p_a, p_0, \alpha, k_{j-1}}_{min, j, n_{j-1}, n_j}  = 
    k_{min,j}(\Providence, p_a, p_0, k_{j-1}, n_{j-1}, n_j)$ such that $$\mathcal{A}(X_j)=\text{Correct}\iff k_j\ge k_{min,j}(\Providence,  \bm{n_j}, p_a, p_0).$$
\end{lemma}
\begin{proof}
    From Definition~\ref{def:minervatwo}, $$\mathcal{A}(X_j)=\text{Correct}\iff \omega_j(k_{j}, k_{j-1}, p_a, p_0, n_j, n_{j-1}) \ge \frac{1}{\alpha}.$$
    Now to apply Lemma~\ref{lemma:imin-exists}, it suffices to show that
    $\omega_j$ is monotone increasing with respect to $k_j$.
    For $j=1$, we have $\omega_1=\tau_1$, so $\omega_1$ is strictly increasing by Lemma \ref{lemma:tau1_increasing}. For $j\ge 2$,
    $$\omega_j(k_j,k_{j-1},p_a,p_0,n_j,n_{j-1},\alpha)=$$$$\sigma(k_{j-1},p_a,p_0,n_{j-1})\cdot \tau_1(k_{j}-k_{j-1},p_a,p_0,n_j-n_{j-1}).$$
    As a function of $k_j$, $\sigma$ is constant, and thus $\omega$ is strictly increasing by Lemma \ref{lemma:tau1_increasing}. Therefore by Lemma \ref{lemma:imin-exists}, we have the desired property.
\end{proof}

\begin{lemma}
\label{lemma:any_ratio_is_sigma_simple}
For $j\ge 1$,
$$\frac{Pr[\bm{K_j}=\bm{k_j} \mid \bm{n_j}, H_a]}{Pr[\bm{K_j}=\bm{k_j} \mid \bm{n_j}, H_0]} = \sigma(k_j, p_a, p_0, n_j).$$
\end{lemma}
\begin{proof}
We induct on the number of rounds.
For $j=1$, we have
$$\frac{Pr[\bm{K_1}=\bm{k_1} \mid \bm{n_1},H_a]}{Pr[\bm{K_1}=\bm{k_1} \mid  \bm{n_1},H_0]} =\frac{Pr[K_1 = k_{1} \mid n_1,H_a]}{Pr[K_1 = k_1 \mid n_1,H_0]} $$$$= \frac{\text{Bin}(k_1,n_1,p_a)}{\text{Bin}(k_1,n_1,p_0)}=\sigma(k_1, p_a, p_0, n_1).$$
Suppose the lemma is true for round $j=m$ with history $\bm{k_m}$.
Observe that
 $$\frac{Pr[\bm{K_{m+1}}=\bm{k_{m+1}} \mid \bm{n_{m+1}},H_a]}{Pr[\bm{K_{m+1}}=\bm{k_{m+1}} \mid \bm{n_{m+1}}, H_0]} $$$$= \frac{ Pr[\bm{K_{m}}=\bm{k_{m}}\mid \bm{n_{m+1}},H_a] \cdot Pr[K_{m+1}'=k_{m+1}'|\bm{k_m},\bm{n_{m+1}},H_a]}{ Pr[\bm{K_{m}}= \bm{k_{m}} \mid  \bm{n_{m+1}},H_0]  \cdot  Pr[K_{m+1}'=k_{m+1}'|\bm{k_m},\bm{n_{m+1}},H_0]  }$$
 $$=\sigma(k_m, p_a, p_0, n_m) \cdot \frac{Pr[K_{m+1}'=k_{m+1}'|\bm{k_m}, \bm{n_{m+1}}, H_a]}{Pr[K_{m+1}'=k_{m+1}'|\bm{k_m},\bm{n_{m+1}},H_0]}$$
 by the induction hypothesis.
Then this is simply equal to
 $$\sigma(k_m, p_a, p_0, n_m)\cdot\frac{\text{Bin}(k_{m+1}',n_{m+1}',p_a)}{\text{Bin}(k_{m+1}',n_{m+1}',p_0)}
 $$$$
 =\frac{p_a^{k_m} (1-p_a)^{n_m-k_m}}{p_0^{k_m} (1-p_0)^{n_m-k_m}} \cdot
 \frac{p_a^{k_{m+1}'} (1-p_a)^{n_{m+1}'-k_{m+1}'}}{p_0^{k_{m+1}'} (1-p_0)^{n_{m+1}'-k_{m+1}'}}
 $$
 $$
 =\sigma(k_{m+1}, p_a, p_0, n_{m+1})
 $$
\end{proof}

\begin{definition} 
\label{def:kmin}
Let $[n_1, \ldots, n_j]$ be the round schedule of an audit that has not stopped by the round $j-1$. Let us define 
\begin{small}
\begin{equation}\label{eq:kMin}
k^{p_a, p_0, \alpha, k_{j-1}}_{min, j, n_{j-1}, n_j}  =
  \min\left\{k : \omega_j(k, k_{j-1},p_a,p_0,n_j, n_{j-1}) \geq \frac{1}{\alpha}  \right\}.
\end{equation}
\end{small}
\end{definition}
As we have seen in Lemma \ref{lemma:minerva2_kmin_exists}, such a value of $k^{p_a, p_0, \alpha, k_{j-1}}_{min, j, n_{j-1}, n_j}$ exists and $k_j \geq k^{p_a, p_0, \alpha, k_{j-1}}_{min, j, n_{j-1}, n_j} $ if and only if the result of the audit is Correct, (\textit{i.e.,} the stopping condition in Definition~\ref{def:minervatwo} holds).

The following lemma shows a Markov-like property of \Providence audit (\textit{i.e.,}
for an audit that has not stopped in the first $j-1$ rounds, only cumulative results of the round $j-1$ matter: cumulative sample size $n_{j-1}$ and the number of ballots for the winner $k_{j-1}$).

\begin{lemma}\label{lemma:markov}
Let $[n_1, \ldots, n_{j-1}, n_j]$ be a round schedule for an execution of  \Providence audit that has not stopped
in any of its first $j-1$ rounds (\textit{i.e.,} for every $i = 1, \ldots, j-1$:
$k_i < k^{p_a, p_0, \alpha, k_{j-1}}_{min, j, n_{j-1}, n_j} $), then: 

\[ 
k^{p_a, p_0, \alpha, k_{j-1}}_{min, j, n_{j-1}, n_j} = k^{p_a, p_0, \alpha, k_{j-1}}_{min, 2, n_{j-1}, n_j} .
\]
\end{lemma}
\begin{proof}
Let $k_{j-1}$ denote the number of ballots drawn for the declared winner up to the round $j-1$ (out of $n_{j-1}$ sampled ballots). The stopping decision for the round $j$ is made as follows:

\[
 k^{p_a, p_0, \alpha, k_{j-1}}_{min, j, n_{j-1}, n_j}  = \min\left\{k : \omega_{j}(k, k_{r-1}, p_a, p_0, n_r, n_{r-1}) \geq \frac{1}{\alpha}  \right\} = 
\]
\[
  =  k^{p_a, p_0, \alpha, k_{j-1}}_{min, 2, n_{j-1}, n_j}  
\]

\end{proof}

That is, the stopping condition is equivalent to that of a two round audit with the same cumulative votes for the winner and cumulative round sizes: the first round is of size $n_{j-1}$ and has $k_{j-1}$ votes for the winner, and the second (cumulative) round size is $n_j$ with $k_j$ (cumulative) votes for the winner. Compare this to the similar property for the $\Bravo$ stopping condition.

\end{document}